\documentclass[11pt,4apaper]{article}
\usepackage[utf8]{inputenc}
\usepackage{verbatim}
\usepackage{amsmath,amssymb,array}
\usepackage{amsthm}
\usepackage[algoruled,linesnumbered]{algorithm2e}
\usepackage{floatflt}
\usepackage{tikz}
\usetikzlibrary{arrows,shapes,through}

\newtheorem{Theorem}{Theorem}
\newtheorem{Lemma}[Theorem]{Lemma}
\newtheorem{Proposition}[Theorem]{Proposition}

\newtheorem{Remark}[Theorem]{Remark}

\newtheorem{Definition}{Definition}
\newtheorem{Problem}{Problem}

\newcommand{\tyb}{\emph{type b }}
\newcommand{\tya}{\emph{type a }}
\newcommand{\grafo}{\mathcal{G}}
\newcommand{\gadgetgrafo}{\mathcal{VG}}

\begin{document}
\title{The $k$-anonymity Problem is Hard}
\author{Paola Bonizzoni\thanks{DISCo, Universit\`a degli Studi di
    Milano-Bicocca, Milano - Italy}
\and Gianluca Della Vedova\thanks{Dipartimento di Statistica,
Universit\`a degli Studi di Milano-Bicocca Milano - Italy}
\and Riccardo Dondi\thanks{Dipartimento di Scienze dei Linguaggi, della Comunicazione e degli Studi Culturali,
Universit\`a degli Studi di Bergamo, Bergamo - Italy}
}

%
%
\maketitle

\begin{abstract}
The problem of publishing personal data without giving up privacy is becoming
increasingly important. An interesting formalization recently proposed is the
$k$-anonymity. This approach requires that
the rows in a table are clustered in sets of size at least $k$ and that
all the rows in a cluster become the same tuple,
after the suppression of some records.
The natural optimization problem, where the goal is to minimize the number of
suppressed entries, is known to be NP-hard
when the  values are over a ternary alphabet, $k=3$ and the rows
length is unbounded.
In this paper we give a lower bound on the approximation factor that any
polynomial-time algorithm can achive on  two restrictions
of the problem, namely (i) when the records values are over a binary alphabet and $k=3$,
and  (ii) when the records have length at most $8$ and $k=4$,
showing that these restrictions of the problem are APX-hard.
\end{abstract}

\section{Introduction}
In many research fields, for example in epidemic analysis, the analysis of
large amounts of personal data is essential.
However, a relevant issue in the management of such data is the
protection of individual privacy.
One approach to deal with such problem is
the $k$-anonymity model
\cite{DBLP:conf/pods/SamaratiS98,Sweeny02,DBLP:journals/tkde/Samarati01},
where a single table is given. The rows of the  table
represent records belonging to different individuals.
Then 
some of the entries in the table are suppressed  so that,
for each record $r$ in the resulting table, there exist at least $k-1$ other records
identical to $r$. At the end of this process, identical rows can be clustered
together; clearly the resulting data is not sufficient to identify each individual.
Different versions of the problem have also been introduced
\cite{Aggar_PODS06}, for example allowing the generalization of  entry values
(an entry value can be replaced with a less specific value) \cite{Aggar_JPT05}.
However, in this paper we will focus only on the suppression model.

A simple parsimonious principle leads to the optimization problem where the
number of entries in the table to be suppressed (or generalized) has to be minimized.
The $k$-anonymity problem is known to be NP-hard for rows of
unbounded length with values over ternary alphabet and $k=3$
\cite{Aggar_ICDT05}. Moreover, a polynomial-time $O(k)$-approximation algorithm
on arbitrary input alphabet, as well as some other  approximation algorithms
for some restricted cases, are known \cite{Aggar_ICDT05}.
Recently, approximation algorithms with factor $O(\log k)$ have been
proposed \cite{Park_SIGMOD07}, even for generalized versions
of the problem \cite{Gionis_ESA07}.

In this paper, we further investigate the approximation and computational
complexity of the $k$-anonymity problem, settling the APX-hardness
for two interesting restrictions of the
problem: (i) when the matrix entries are over a \emph{binary} alphabet and
$k=3$, or (ii) when the matrix has $8$ columns and $k=4$.
We notice that these are the first inapproximability results
for the $k$-anonymity problem. More precisely, we prove the two inapproximability results
by designing two
$L$-reductions \cite{Ausiello-Crescenzi} from the Minimum Vertex
Cover problem to $3$-anonymity problem over binary alphabet and
$4$-anonymity problem when the rows are of length $8$ respectively.
Those two restrictions are of particular interests as some data can be inherently
binary (e.g. gender) and publicly disclosed usually have only a few
columns, therefore solving such restrictions could help for most practical cases.

The rest of the paper is organized as follows.
In Section \ref{pre-sec} we introduce some preliminary definition and we give
the formal definition of the $k$-anonymity problem. In Section \ref{ABT-comp} we
show that the $3$-anonymity is APX-hard, even when the matrix is restricted to binary data,
while in Section \ref{sec-APX-8-4} we show that the $4$-anonymity problem is APX-hard, even
when the rows have length bounded by $8$.

\section{Preliminary Definitions}
\label{pre-sec}


In this section we introduce some preliminary definitions that will be used in
the rest of the paper.
A graph $\grafo=(V,E)$ is cubic when each vertex in $V$ has degree three.


Given an alphabet $\Sigma$, a row $r$ is a vector of
elements taken from the set $\Sigma$, and the
$j$-th element of $r$ is denoted by $r[j]$.
Let $r_1, r_2$ be two equal-length rows. Then $H(r_1,r_2)$ is the Hamming
distance of $r_1$ and $r_2$, i.e. $|\{i: r_1[i]\neq r_2[i]\}|$.
Let $R$ be a set of $l$ rows, then
a \emph{clustering} of $R$ is a partition $P=(P_1,\dots , P_t )$ of $R$.
Since all rows in a table have the same number of elements and the order of
the elements of a rows is important, we may think of a row over the set
$\Sigma$ as a string over alphabet $\Sigma$.

Given a clustering $P=(P_1,\dots , P_t )$ of $R$, we define the \emph{cost}
of the row $r$ belonging to a set $P_i$, as
$|\{j:\exists r_1,r_2\in P_i,\ r_1[j]\neq r_2[j]\}|$, that is the number of entries
of $r$ that have to be supressed so that all rows in $P_i$ are identical.
Similarly we define the cost of a set $P_i$,  denoted by $c(P_i)$, as
$|P_i||\{j:\exists r_1,r_2\in P_i,\ r_1[j]\neq r_2[j]\}|$.
The cost of $P$, denoted by $c(P)$, is defined as
$\sum_{P_i\in P} c(P_i)$.
Notice that, given a clustering $P=(P_1,\dots, P_t)$ of $R$, the quantity
$|P_i|\max_{r_1,r_2\in P_i}\{H(r_1,r_2)\}$ is a lower bound for $c(P_i)$,
since all the positions for which $r_1$ and $r_2$ differ will
be deleted in each row of $P_i$.
We are now able to formally define the $k$-Anonymity Problem
($k$-AP) as follows:

\begin{Problem}$k$-AP.\\
\textbf{Input}: a set $R$ of rows over an alphabet $\Sigma$.\\
\textbf{Output}: a clustering $P=(P_1,\dots, P_t)$ of $R$
such that for each set $P_i$, $|P_i|\geq k$\\
\textbf{Goal}: to minimize $c(P)$.
\end{Problem}

%
The following Property will be used in several proofs.

\begin{Proposition}\cite{Aggar_ICDT05}
\label{Property:cluster-dim}
Let $R$ be an instance of $k$-AP, and let $P$ be a solution of $k$-AP over instance $R$.
Then we can compute in polynomial time a solution $P'$, with $c(P') \leq c(P)$, such that
for each cluster $P'_i$ of $P'$, $k \leq |P'_i| \leq 2k-1$.
\end{Proposition}

We will study two restrictions of the $k$-anonymity
problem. In the first restriction, denoted by
$3$-ABP, the rows are over a binary alphabet $\Sigma=\{0_b, 1_b \}$ and
$k=3$. %
In the second restriction, denoted by $4$-AP($8$), $k=4$ and the rows are over an
arbitrary alphabet and have length $8$.

In the remaining of the paper we will prove the APX-hardness of both
restrictions, presenting two different reductions from the
Minimum Vertex Cover on Cubic Graphs (MVCC) problem, which is known to be APX-hard \cite{AK92}.
Consider a cubic graph $\grafo=(V,E)$, where $|V|=n$ and $|E|=m$, the MVCC
problem asks for a subset
$C \subseteq V$ of minimum cardinality, such that for each edge
$(v_i,v_j) \in E$, at least one of $v_i$ or $v_j$ belongs to $C$.

\section{APX-hardness of $3$-ABP}
\label{ABT-comp}

In this section we will show that $3$-ABP
is APX-hard via an L-reduction from Minimum Vertex Cover on Cubic
Graphs (MVCC), which is known to be APX-hard \cite{AK92}.
From Proposition \ref{Property:cluster-dim}, it follows
Remark \ref{Remark:cluster-dim},
that shows that we can restrict ourselves to
solutions of  $3$-ABP where each cluster contains at most $5$ rows.

\begin{Remark}
\label{Remark:cluster-dim}
Let $R$ be an instance of $3$-ABP, and let $P$ be a solution of $3$-ABP over instance $R$.
Then we can compute in polynomial time a solution $P'$, with $c(P') \leq c(P)$, such that
for each cluster $P'_i$ of $P'$, $3 \leq |P'_i| \leq 5$.
\end{Remark}

Let $\grafo=( V , E )$ be an instance of MVCC, the reduction builds an
instance $R$ of $3$-ABP associating with each vertex $v_i \in V$
a set of  rows $R_i$, and with each $e=(v_i, v_j) \in E$ a row
$r_{i,j}$.
\begin{figure}[htb!]
\begin{center}
\begin{tikzpicture}[label distance=2mm , scale=.8,every node/.style={inner sep=0pt, minimum width=2pt}]
  \coordinate (x)   at (3,5);
  \coordinate (y) at (0,0);
  \coordinate (z)      at (6,0);
  \node (i1) [label=-45:$c_{i,1}$, pin distance=2cm] at (y) {};
  \node (i2) [label=180:$c_{i,2}$] at (x) {};
  \node (i3) [label=225:$c_{i,3}$] at (z) {};
  \node (i4) [label=150:$c_{i,4}$] at (barycentric cs:x=1,y=1 ,z=0)    {};
  \node (i5) [label=0:$c_{i,5}$] at (barycentric cs:x=1,y=0 ,z=1)    {};
  \node (i6) [label=90:$c_{i,6}$] at (barycentric cs:x=1,y=1 ,z=1)    {};
  \node (i7) [label=270:$c_{i,7}$] at (barycentric cs:x=0,y=1 ,z=1)    {};

  \node (Ji1) [label=180:$J_{i,1}$] at (barycentric cs:x=-0.08,y=1 ,z=-0.08)    {};
  \node (Ji2) [label=90:$J_{i,2}$] at (barycentric cs:x=1,y=-0.08 ,z=-0.08)    {};
  \node (Ji3) [label=0:$J_{i,3}$] at (barycentric cs:x=-0.08,y=-0.08 ,z=1)    {};

  \draw  (x.south) -- (y.north east) -- (z.north west) -- cycle;

  \draw  (i4) -- (i6);
  \draw  (i5) -- (i6);
  \draw  (i7) -- (i6);

  \foreach \v/\w in
  {Ji1/i1, Ji2/i2, Ji3/i3}{

    \draw  (\v) to [bend right=45] (\w);
    \draw  (\v) to [bend right=15] (\w);
    \draw  (\v) to [bend left=45] (\w);
    \draw  (\v) to [bend left=15] (\w);
  }

  \foreach \v in
      {i1, i2, i3, i4, i5, i6, i7, Ji1, Ji2, Ji3}{
         \fill (\v) circle(2.5pt);
    }

  \coordinate (xj)   at (13,5);
  \coordinate (yj) at (10,0);
  \coordinate (zj)      at (16,0);
  \node (j1) [label=-45:$c_{j,1}$, pin distance=2cm] at (yj) {};
  \node (j2) [label=180:$c_{j,2}$] at (xj) {};
  \node (j3) [label=225:$c_{j,3}$] at (zj) {};
  \node (j4) [label=150:$c_{j,4}$] at (barycentric cs:xj=1,yj=1 ,zj=0)    {};
  \node (j5) [label=0:$c_{j,5}$] at (barycentric cs:xj=1,yj=0 ,zj=1)    {};
  \node (j6) [label=90:$c_{j,6}$] at (barycentric cs:xj=1,yj=1 ,zj=1)    {};
  \node (j7) [label=270:$c_{j,7}$] at (barycentric cs:xj=0,yj=1 ,zj=1)    {};

  \node (Jj1) [label=180:$J_{j,1}$] at (barycentric cs:xj=-0.08,yj=1 ,zj=-0.08)    {};
  \node (Jj2) [label=90:$J_{j,2}$] at (barycentric cs:xj=1,yj=-0.08 ,zj=-0.08)    {};
  \node (Jj3) [label=0:$J_{j,3}$] at (barycentric cs:xj=-0.08,yj=-0.08 ,zj=1)    {};

  \draw  (xj.south) -- (yj.north east) -- (zj.north west) -- cycle;

  \draw  (j4) -- (j6);
  \draw  (j5) -- (j6);
  \draw  (j7) -- (j6);

  \foreach \v/\w in
  {Jj1/j1, Jj2/j2, Jj3/j3}{

    \draw  (\v) to [bend right=45] (\w);
    \draw  (\v) to [bend right=15] (\w);
    \draw  (\v) to [bend left=45] (\w);
    \draw  (\v) to [bend left=15] (\w);
  }

  \foreach \v in
      {j1, j2, j3, j4, j5, j6, j7, Jj1, Jj2, Jj3}{
         \fill (\v) circle(2.5pt);
    }

    \draw  (j1) to [auto,swap,bend right=90] node {$EG_{ij}$}  (i3);

\end{tikzpicture}
  \caption{Gadgets for $v_i$, $v_j$, $(v_i, v_j)$}
  \label{fig:graph-tab}
\end{center}
\end{figure}
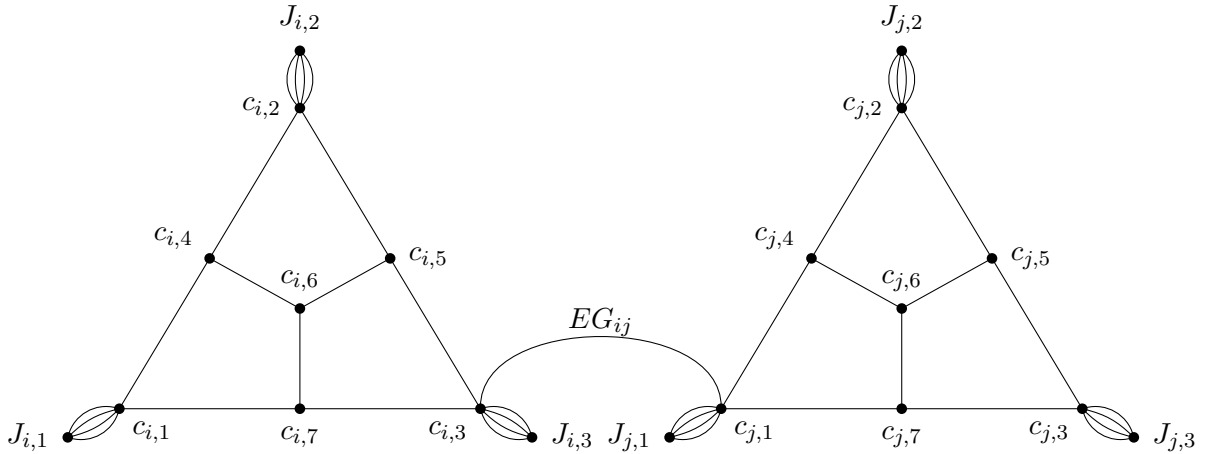
Actually, 
starting from the cubic graph $\grafo$, the reduction builds an intermediate
multigraph, denoted as \emph{gadget graph} $\gadgetgrafo$ 
-- a snippet
of a gadget graph obtainable through our reduction is represented in
Fig.~\ref{fig:graph-tab}.
The reduction associates with each vertex $v_i$ of $\grafo$
a vertex gadget $VG_i$ containing a \emph{core}
vertex gadget $CVG_i$ and some other vertices and edges called respectively
jolly vertices and  jolly edges.
More precisely,
the vertex-set of a core vertex gadget $CVG_i$ consists of the seven vertices
$c_{i,1}, c_{i,2}, c_{i,3}, c_{i,4}, c_{i,5}, c_{i,6}, c_{i,7}$.
The vertices $c_{i,1}$, $c_{i,2}$ and $c_{i,3}$ of $CVG_i$ are
called \emph{docking vertices}.
The edge-set of $CVG_i$ consists of nine edges between vertices of $CVG_i$
(see Fig.~\ref{fig:graph-tab}). Such a set of edges is defined as the set of
\emph{core edges} of $VG_i$.
The vertex-set of a vertex gadget consists of the seven vertices of $CVG_i$
and of three more vertices $J_{i,1}, J_{i,2}, J_{i,3}$, called \emph{jolly vertices}
of $VG_i$.
The edge-set of $VG_i$ consists of the edge-set of $CVG_i$
and of three
sets of four parallel edges (see Fig.~\ref{fig:graph-tab}).
More precisely,
for each docking vertex $c_{i,z}$ adjacent to a jolly vertex $J_{i,z}$,
we define a set $E^J_{i,z}$ of four parallel edges between $c_{i,z}$ and $J_{i,z}$.
The set of edges $E^J_i= \bigcup_{z \in \{1,2,3 \} } E^J_{i,z}$ is called the
set of \emph{jolly edges} of $VG_i$.


Each edge $(v_i,v_j)$ of $\grafo$ is encoded by an edge gadget $EG_{ij}$ consisting
of a single edge that
connects a docking vertex of $VG_i$ with one of $VG_j$, so that
in the resulting graph each docking vertex is an endpoint of exactly one edge
gadget (this can be achieved trivially as the original graph is cubic.)
The resulting graph, denoted by $\gadgetgrafo$, is called \emph{gadget graph}.
An edge gadget is said to be \emph{incident} on a
vertex gadget $VG_i$ if it is incident on a docking vertex of $VG_i$.
In our reduction we will associate a row with each edge of the graph
gadget. Therefore $3$-ABP is equivalent to partitioning the edge set of the
gadget graph into sets of at least three edges. Hence in what follows
we may use edges of $\gadgetgrafo$ to denote the corresponding rows.
Before giving some details, we present an overview of the reduction.

First, the input set $R$ of rows is defined, so that each row corresponds to an edge
of the gadget graph.
Then, it is shown that, starting from a general solution, we can restrict ourselves to a canonical solution,
where there exist only two possible partitions of the rows of a vertex gadget
(and possibly some edge gadgets). Such partitions are denoted as
\tya and \tyb solution.
Finally, the rows of a vertex gadget that belongs to a \tyb (\tya resp.) solution are related to vertices
in the cover (not in the cover, respectively)  of the graph $\grafo$.


We are now able to introduce our reduction.
All the rows in $R$ are the juxtaposition of $n+2$ blocks, where the $i$-th block,
for $1 \leq i \leq n$, is associated with vertex $v_i \in V$, the
$(n+1)$-th block is called \emph{jolly block}, and
the $(n+2)$-th block is called \emph{edge block}. The first $n$ blocks are
called \emph{vertex blocks}, and each vertex block has size $21$.
The jolly block has size $6n$, and the edge block has size $3n$.


The rows associated with edges of the  gadget graph $\gadgetgrafo$
are obtained by introducing the following operations on rows (also called
encoding operations). For simplicity's
sake we will use a string-based notation.

\begin{Definition}[Encoding operations]
Let $VG_i$ be a vertex gadget, $CVG_i$ be a core vertex gadget, $c_{i,j}$ be a
vertex of $CVG_i$, $1 \leq i \leq n$, $1 \leq j \leq 7$ , and let $r$ be a
row. Then
the  {\em vertex encoding} of $c_{i,j}$  applied to $r$,
denoted by $v\text{-}enc_{i,j}(r)$, is obtained by assigning $1_b$ to the
positions $3j-2$, $3j-1$, $3j$ of the $i$-th block of $r$ (and leaving all
other entries as in $r$).
The  {\em gadget encoding} of $VG_i$ applied to $r$, denoted by
$g\text{-}enc_i(r)$, is obtained by assigning $1_b$ to the
positions $3i-2$, $3i-1$, $3i$ of the edge block of $r$ (and leaving all
other entries as in $r$).
Finally, let $J_{i,x}$ be a jolly vertex of $VG_i$, $1 \leq i \leq n$, $1 \leq
x \leq 3$, and let $r$ be a row, then
the {\em jolly encoding} $j\text{-}enc_{i,x}$ of $J_{i,x}$ applied to $r$,
denoted by $j\text{-}enc_{i,x}(r)$, is obtained by assigning $1_b$ to the
to the positions $6(i-1)+ x$, $6(i-1)+ x+1$ of the jolly block of row $r$.
\end{Definition}

\begin{table}[hbt]
  \centering
  \begin{tabular}{|l|p{10em}|p{10em}|p{11em}|}\hline%
    Operation&Positions of the $i$-th vertex blocks set to $1_b$ &Positions of the edge
    blocks set to $1_b$&Positions of the jolly block set to $1_b$\\\hline
$v\text{-}enc_{i,j}(r)$&$3j-2$, $3j-1$, $3j$&&\\
$g\text{-}enc_i(r)$&&$3i-2$, $3i-1$, $3i$&\\
$j\text{-}enc_{i,x}(r)$&&&$6(i-1)+ x$, $6(i-1)+ x+1$\\\hline
  \end{tabular}
  \caption{Summary of the encoding operations}
  \label{tab:operations}
\end{table}

Notice that the vertex encoding and the gadget encoding operations set to
$1_b$ at most $3$ entries of any row, while the jolly encoding operation
sets to $1_b$ at most $2$ entries of any row.
Let $c_{i,x}$ be a docking vertex of $VG_i$,
and let $c_{i,y}$, $c_{i,z}$ be the two core vertices of $VG_i$ adjacent
to $c_{i,x}$,
let $(c_{i,x},c_{i,y})$  be a \emph{core edge},
and let  $J_{i,x}$ be a jolly vertex adjacent to $c_{i,x}$.
Then the row
$r_{i,x,y}$ associated  with $(c_{i,x},c_{i,y})$ is
$g\text{-}enc_i \left( v\text{-}enc_{i,y} \left(
    v\text{-}enc_{i,x}\left(0_b^{30n}\right)\right)\right)$,
each row associated with a jolly edge $(c_{i,x}, J_{i,x})$, denoted by
$r_{i,x,y,z}$ (for $1\le z\le 4$) is
$j\text{-}enc_{i,x}\left(
  g\text{-}enc_i\left(v\text{-}enc_{i,z}\left(v\text{-}enc_{i,y}\left(
v\text{-}enc_{i,x}\left(0_b^{30n}\right)\right)\right)\right) \right)$,
Each row associated with a jolly edge is called
\emph{jolly row} and the set of the 4 jolly rows incident
to  vertex $c_{i,x}$ is called jolly row set of $c_{i,x}$.
Finally, let $EG_{ij}=(c_{i,x},c_{j,y})$ be an \emph{edge gadget}
The row $r_{i,j,x,y}$ associated with $EG_{ij}$ is
$g\text{-}enc_j\left(g\text{-}enc_i\left(v\text{-}enc_{j,y}\left(
      v\text{-}enc_{i,x}\left(0_b^{30n}\right)\right)\right)\right)$.
In Table~\ref{tab:encodings} the rows associated with the various edge are summarized.

\begin{table}[hbt]
  \centering
  \begin{tabular}{|l|l|}\hline%
    Edge&Associated row\\\hline
Core edge $(c_{i,x},c_{i,y})$&$g\text{-}enc_i \left( v\text{-}enc_{i,y} \left(
    v\text{-}enc_{i,x}\left(0_b^{30n}\right)\right)\right)$\\
Jolly edge $(c_{i,x}, J_{i,x})$&%
$j\text{-}enc_{i,x}\left(
  g\text{-}enc_i\left(v\text{-}enc_{i,z}\left(v\text{-}enc_{i,y}\left(
v\text{-}enc_{i,x}\left(0_b^{30n}\right)\right)\right)\right) \right)$\\
Edge gadget $EG_{ij}$&%
$g\text{-}enc_j\left(g\text{-}enc_i\left(v\text{-}enc_{j,y}\left(
      v\text{-}enc_{i,x}\left(0_b^{30n}\right)\right)\right)\right)$\\\hline
  \end{tabular}
  \caption{Encodings of the edges}
  \label{tab:encodings}
\end{table}

For example consider the row $r_{i,1,4}$ associated with the core edge $(c_{i,1}, c_{i,4})$.
Observe that $v\text{-}enc_{i,1}$ sets to $1_b$ 
the first three positions of the $i$-th block of $r_{i,1,4}$, while $v\text{-}enc_{i,4}$
sets to $1_b$ 
the positions $10$, $11$, $12$ of the $i$-th block of $r_{i,1,4}$.
Finally, $g\text{-}enc_{i}$ sets to $1_b$ the positions $3i-2$, $3i-1$ and $3i$ of the edge block of $r_{i,1,4}$.
Edge $(c_{i,1}, c_{i,4})$ is associated with the following row $r_{i,1,4}$:

\[
\underbrace{0_b0_b \dots 0_b}_{block \ 0} \dots \underbrace{1_b1_b1_b\ 0_b0_b0_b\ 0_b0_b0_b\ 1_b1_b1_b\ 0_b0_b0_b \dots}_{block \ i} \dots
\underbrace{0_b0_b \dots 0_b}_{jolly\ block }
\underbrace{0_b0_b \dots 1_b1_b1_b \dots 0_b0_b0_b}_{edge\ block }.
\]
Observe that by construction only jolly rows may have a $1_b$ in a position
of the jolly block. It is immediate to notice that clustering together three
or more jolly rows associated with parallel edges has cost $0$.
We recall that
we may use edges of $\gadgetgrafo$ to denote the corresponding rows.

\begin{Proposition}
\label{prop:distanze}
Let  $e_1, e_2$ be two edges of $CVG_i$, let $e_3$ be an edge of $VG_j$ (with
$i\neq j$), let $e_j$ be a jolly edge of $VG_i$,
let $e_5$ be a jolly edge of $VG_z$, and let $EG_{ix}$, $EG_{jl}$
be two  edge gadgets. Then:
\begin{enumerate}
\item
$H(e_1,e_3)\ge 18$;
\item
$H(EG_{ix},e_j),H(EG_{jl},e_j)\ge 14$;
\item
If $e_1$ and $e_2$ are incident on the same
  vertex, then $H(e_1, e_2)=6$;
\item
If $e_1$, $e_2$  are not incident on the same
  vertex, then $H(e_1, e_2)=12$;
\item
If $e_1$ and $e_j$ are incident on the same
  vertex, then $H(e_1,e_j)=5$;
\item
If $e_1$ and $e_j$ are not incident on the same
  vertex, then $H(e_1,e_j)\ge 11$;
\item
If $e_1$ and $EG_{ix}$ are incident on the same vertex, then $H(EG_{ix}, e_1)=9$;
\item
If $e_1$ and $EG_{ix}$ are not incident on the same vertex, then
  $H(EG_{ix}, e_1)\geq 15$;
\item
$H(EG_{ix},EG_{jl})\geq 18$;
\item
If $i$, $x$, $j$, $l$ are all distinct (i.e. the two edge
gadgets are not incident on the same vertex gadget), then $H(EG_{ix},EG_{jl})=
24$.
\item
If $e_j$, $e_5$ are not in
the same jolly set, then $H(e_j, e_5) \geq 12$.
\item
Let $r$ be a row not incident on a common vertex with  $e_j$. Then $H(e_j,r)\ge 11$.
\end{enumerate}
\end{Proposition}
\begin{proof}
Observe that all the cases can be easily proved
by observing that each row is obtained applying
3 kinds of encoding operations, where the vertex encoding and gadget
encodings assign values $1_b$  in three positions, while the jolly encoding
assigns $1_b$ into two positions. We now prove the various cases, following
the order of the statement.

\begin{enumerate}
\item
Since $e_1$ and $e_3$ are incident on different vertices, the
associated rows are the result of applying once the encoding operations with
different values of $i$. Therefore the positions where a $1_b$ is set are
disjoint. Since there are at least $12$ such positions of the vertex blocks
and  $6$ positions of the edge block, we obtain $H(e_1,e_3)\ge 18$.
\item
Notice that there are $2$ positions of the jolly block that are set to $1_b$
as a result of applying the jolly encoding to $e_j$, while the whole block is
set to $0_b$ in $EG_{i,x}$.
Since $EG_{i,x}$ is subject to two gadget encodings, while $e_j$ is subject
only to one gadget encoding, $EG_{i,x}$ and $e_j$ have $3$ different entries
also in an edge block.
Moreover at most two of the overall five vertex encoding operations have the
same arguments (those corresponding to a shared docking vertex), resulting in
an additional $9$ different entries.
\item
Since $e_1$ and $e_2$ are incident on a common vertex, they share a gadget
encoding operation and a vertex encoding operation, therefore there are two
different vertex encoding operations that result in $6$ different entries.
\item
Since $e_1$ and $e_2$ are not incident on a common vertex, but are in the same
vertex gadget, they share only a gadget
encoding operation, therefore there are four
different vertex encoding operations that result in $12$ different entries.
\item
Since $e_1$ and $e_j$ are incident on a common vertex, they share a gadget
encoding operation and two vertex encoding operations, while they differ for a
jolly encoding operation and a vertex encoding operation,  resulting in $5$ different entries.
\item
Since $e_1$ and $e_j$ are not incident on a common vertex, they share a gadget
encoding operation (since by hypothesis $e_1$ and $e_j$ are in the same vertex gadget),
and at most one vertex encoding operations (if $e_1$ is incident on a vertex adjacent
to the docking vertex on which $e_j$ is also incident), while they differ for a
jolly encoding operation and three vertex encoding operations,  resulting in
at least $11$ different entries.
\item
Since $e_1$ and $EG_{ix}$ are incident on a common (docking) vertex,
they share a gadget encoding operation and a vertex encoding operation,
while they differ for a gadget encoding operation and two vertex encoding
operations,  resulting in $9$ different entries.
\item
Since $e_1$ and $EG_{ix}$ are not incident on a common  vertex,
they share a gadget encoding operation  (since $EG_{ix}$ is incident on a
docking vertex of $VG_i$),
while they differ for a gadget encoding operation and four vertex encoding
operations,  resulting in $15$ different entries.
\item
Since $EG_{ix}$ and $EG_{jl}$ are not incident on a common vertex,
they might share a gadget
encoding operation (if the two edge gadget are incident on the the same vertex
gadget), while the differ for
four vertex encoding operations and two gadget encoding operations, resulting
in at least $18$ different entries.
\item
Since $EG_{ix}$ and $EG_{jl}$ are not incident on a common  vertex gadget,
they  share no encoding operations, resulting
in at $24$ different entries.
\item
Since $e_j$  and $e_5$  are not in the same jolly set,
they might share a gadget encoding operation and a vertex encoding operations
(if the two jolly edges are in  the same vertex gadget), while the differ for
four vertex encoding operations, resulting
in at least $12$ different entries.
\item
The results follows from the previous cases (case 2, 6 and 11),
as row $r$ is either a row of
a $CVG_z$, for some $z$, a jolly row in a jolly row set of a different vertex of $\gadgetgrafo$, or an edge gadget.
\end{enumerate}
\end{proof}

The cost of a solution $S$ is specified by introducing the
notion of virtual cost of a single row $r$ of $R$.  Let $S$ be a solution of 3-ABP,
and let $C$ be the cluster of $S$ to which $r$ belongs. 
Let  $r$ be a non-jolly row, we define the \emph{virtual
cost} of $r$ in the solution $S$, denoted as $virt_S(r)$, as the cost
of $C$ divided by the number of non-jolly rows in $C$. Otherwise,
if $r$ is a jolly row, then $virt_S(r)=0$.
Given the above notion, observe that the cost $c(C)$ of set $C$ is equal to
$\sum_{r\in C}virt_S(r)$ and that for a solution $S$, the cost $c(S)$ of set $S$
is equal to $\sum_{r\in R}virt_S(r)$.

In the following we will consider only \emph{canonical solutions} of $3$-ABP,
that is solutions where the rows for
each vertex gadget $VG_i$ and edge gadgets eventually incident on
$VG_i$  are clustered into \emph{type a} and \emph{type b}
solutions constructed as follows.

The \tya solution defines the partition of the rows
for vertex gadget $VG_i$ and consists of
six clusters: three clusters of rows of $CVG_i$,  each one  is  made of the
three edges incident on vertex $v$, where $v$ is one of the three
vertices $ c_{i,4}$, $c_{i,5}$ and $c_{i,7}$, and
three more clusters, each one consisting of the jolly rows associated with
one of the three docking vertices of $VG_i$.
%

The \tyb solution defines the partition of the rows
for a vertex gadget $VG_i$   and some edge gadgets incident on
$VG_i$. It consists of four clusters
containing rows of $CVG_i$. One  of them  consists of the three
edges incident on $c_{i,6}$. The remaining three clusters are associated
with the three docking vertices of $VG_i$. For each docking vertex
$c_{i,x}$, the cluster associated with $c_{i,x}$ consists of the two core
edges of $CVG_i$ that are incident on $c_{i,x}$, together with either
the edge gadget incident on $c_{i,x}$ or one jolly edge incident in $c_{i,x}$.
Finally, there are three more clusters, each one consisting of all
remaining jolly edges associated with parallel edges incident on
one of  the three docking vertices of $VG_i$.
Notice that in a
\tyb solution each cluster associated with a docking vertex may
contain an edge gadget or not, the only requirement is that at
least one of the clusters contains an edge gadget.
Notice that \tya and \tyb solutions cluster together edges incident on a
common vertex (by an abuse of language, we will call canonical such a
cluster): the common vertex of a canonical cluster is called the \emph{center}
of the cluster.

\begin{figure}
\begin{center}
\begin{tikzpicture}[label distance=2.6mm , scale=.8,every node/.style={inner sep=0pt, minimum width=2pt}]
  \coordinate (x)   at (90:3cm);
  \coordinate (y) at (210:3cm);
  \coordinate (z)      at (-30:3cm);
  \node (i1) [label=-20:$c_{i,1}$, pin distance=2cm] at (y) {};
  \node (i2) [label=176:$c_{i,2}$] at (x) {};
  \node (i3) [label=200:$c_{i,3}$] at (z) {};
  \node (i4) [label=150:$c_{i,4}$] at (barycentric cs:x=1,y=1 ,z=0)    {};
  \node (i5) [label=0:$c_{i,5}$] at (barycentric cs:x=1,y=0 ,z=1)    {};
  \node (i6) [label=-2:$c_{i,6}$] at (barycentric cs:x=1,y=1 ,z=1)    {};
  \node (i7) [label=270:$c_{i,7}$] at (barycentric cs:x=0,y=1 ,z=1)    {};

  \node (Ji1) [label=180:$J_{i,1}$] at (barycentric cs:x=-0.12,y=1 ,z=-0.2)    {};
  \node (Ji2) [label=90:$J_{i,2}$] at (barycentric cs:x=1,y=-0.15 ,z=-0.15)    {};
  \node (Ji3) [label=0:$J_{i,3}$] at (barycentric cs:x=-0.12,y=-0.2 ,z=1)    {};

  \draw  (x.south) -- (y.north east) -- (z.north west) -- cycle;

  \draw  (i4) -- (i6);
  \draw  (i5) -- (i6);
  \draw  (i7) -- (i6);

  \foreach \v/\w in
  {Ji1/i1, Ji2/i2, Ji3/i3}{

    \draw  (\v) to [bend right=45] (\w);
    \draw  (\v) to [bend right=15] (\w);
    \draw  (\v) to [bend left=45] (\w);
    \draw  (\v) to [bend left=15] (\w);
  }

  \foreach \v in
      {i1, i2, i3, i4, i5, i6, i7, Ji1, Ji2, Ji3}{
         \fill (\v) circle(2.5pt);
    }

    \begin{scope}[every to/.style={draw,dashed}]
    \draw  (Ji1) to [bend right=80] (i1);
    \draw  (Ji1) to [bend left=80] (i1);
    \draw  (Ji2) to [bend right=80] (i2);
    \draw  (Ji2) to [bend left=80] (i2);
    \draw  (Ji3) to [bend right=80] (i3);
    \draw  (Ji3) to [bend left=80] (i3);

    \draw  (i1) to (i6);
    \draw  (i2) to (i6);
    \draw  (i3) to (i6);

    \draw  (i1) to [bend left=60] (i2);
    \draw  (i2) to [bend left=60] (i3);
    \draw  (i1) to [bend right=60] (i3);

    \end{scope}
\end{tikzpicture}
\caption{A \tya solution for the rows associated with $VG_i$, where the dashed
  lines represent  borders among clusters.
Recall that each edge $e$ of $VG_i$ corresponds to a row.}
\label{fig:ty_a-ty_a}
\end{center}
\end{figure}
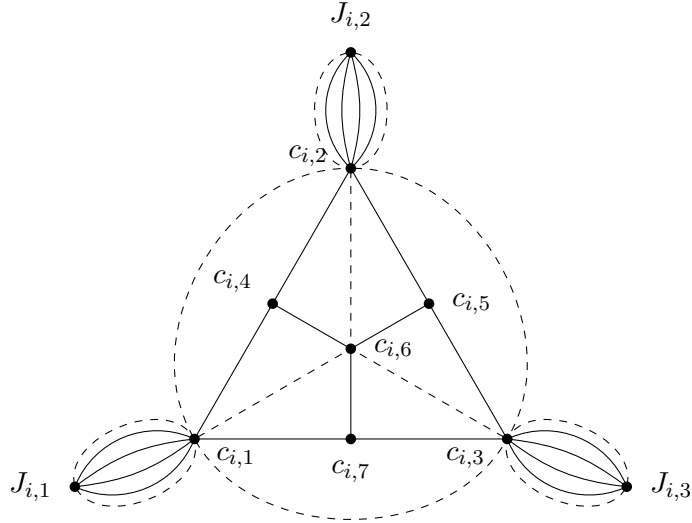

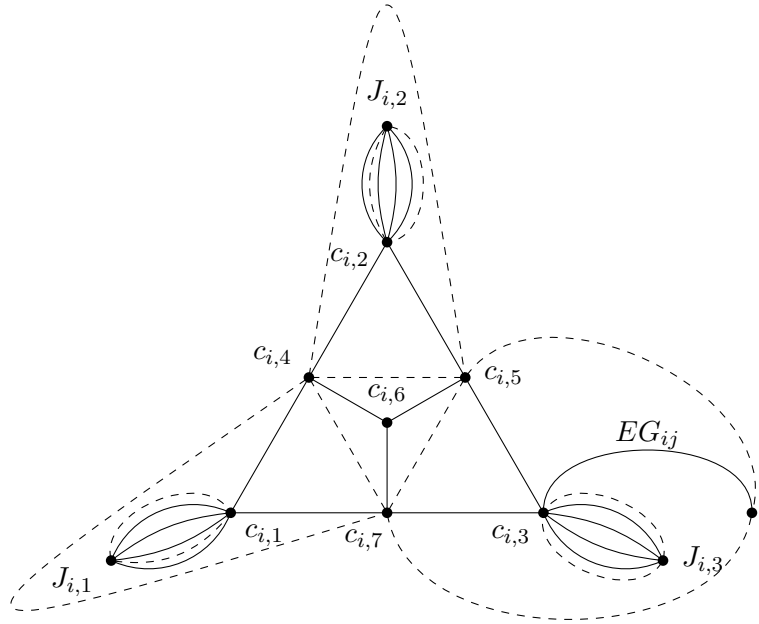
\begin{figure}
\begin{center}
\begin{tikzpicture}[  label distance=2mm , scale=.8,every node/.style={inner sep=0pt, minimum width=2pt}]
  \coordinate (x)   at (90:3cm);
  \coordinate (y) at (210:3cm);
  \coordinate (z)      at (-30:3cm);
  \node (i1) [label=-40:$c_{i,1}$, pin distance=2cm] at (y) {};
  \node (i2) [label=200:$c_{i,2}$] at (x) {};
  \node (i3) [label=225:$c_{i,3}$] at (z) {};
  \node (i4) [label=150:$c_{i,4}$] at (barycentric cs:x=1,y=1 ,z=0)    {};
  \node (i5) [label=0:$c_{i,5}$] at (barycentric cs:x=1,y=0 ,z=1)    {};
  \node (i6) [label=90:$c_{i,6}$] at (barycentric cs:x=1,y=1 ,z=1)    {};
  \node (i7) [label=260:$c_{i,7}$] at (barycentric cs:x=0,y=1 ,z=1)    {};

  \node (Ji1) [label=200:$J_{i,1}$] at (barycentric cs:x=-0.12,y=1 ,z=-0.2)    {};
  \node (Ji2) [label=90:$J_{i,2}$] at (barycentric cs:x=1,y=-0.15 ,z=-0.15)    {};
  \node (Ji3) [label=0:$J_{i,3}$] at (barycentric cs:x=-0.12,y=-0.2 ,z=1)    {};

  \node (VGj) at (barycentric cs:x=0,y=-0.4 ,z=1)    {};

  \draw  (x.south) -- (y.north east) -- (z.north west) -- cycle;

  \draw  (i4) -- (i6);
  \draw  (i5) -- (i6);
  \draw  (i7) -- (i6);

  \foreach \v/\w in
  {Ji1/i1, Ji2/i2, Ji3/i3}{

    \draw  (\v) to [bend right=45] (\w);
    \draw  (\v) to [bend right=15] (\w);
    \draw  (\v) to [bend left=45] (\w);
    \draw  (\v) to [bend left=15] (\w);
  }

  \foreach \v in
      {i1, i2, i3, i4, i5, i6, i7, Ji1, Ji2, Ji3, VGj}{
         \fill (\v) circle(2.5pt);
    }

    \draw  (VGj) to [auto,swap,bend right=90] node {$EG_{ij}$}  (i3);

    \begin{scope}[every to/.style={draw,dashed}]
      \foreach \v/\w in
      {i5/i4, i4/i7, i5/i7}{
        \draw  (\v) to (\w);
      }
    \draw  (Ji1) to [bend right=30] (i1);
    \draw  (Ji1) to [bend left=80] (i1);
    \draw  (Ji2) to [bend right=30] (i2);
    \draw  (Ji2) to [bend left=80] (i2);
    \draw  (Ji3) to [bend right=80] (i3);
    \draw  (Ji3) to [bend left=80] (i3);

    \draw  (i4) [dashed] .. controls (90:9cm)  .. (i5);
    \draw  (i4) [dashed] .. controls (206:9cm)  .. (i7);

    \draw  (i5) to [bend left=80] (VGj);
    \draw  (i7) to [bend right=80] (VGj);

    \end{scope}
\end{tikzpicture}
\caption{A \tyb solution for the rows associated with $VG_i$ and $EG_{ij}$, where the dashed
  lines represent  borders among clusters.
Recall that each edge $e$ corresponds to a row.}
\label{fig:ty_a-ty_b}
\end{center}
\end{figure}

\begin{Proposition}
\label{prop-cost-canonical-sol}
Let $S$ be a canonical solution of an instance of $3$-ABT associated with an
instance of MVCC, and let $VG_i$, $VG_j$ be two vertex gadgets
such that the rows of  $VG_i$ are clustered in a \tya solution in $S$ and
rows of $VG_j$ are clustered in a \tyb solution in $S$.
Then each edge gadget has a virtual
cost of $12$ in $S$, the rows of $VG_i$ have a
total cost of $81$, while the rows of $VG_j$ have a total cost of $99$.
\end{Proposition}
\begin{proof}
Let $EG_{ij}$ be an edge gadget of $\gadgetgrafo$. Observe that in a canonical solution each edge gadget belongs
to a \tyb solution. Consider a \tyb solution for $VG_i$ containing $EG_{ij}$. By definition
of \tyb solution, $EG_{ij}$ is co-clustered with two rows of $CVG_i$, so that those rows are incident
on a common docking vertex with $EG_{ij}$. It follows that $12$ entries
($9$ of the vertex blocks, $3$ of the edge block)
are deleted in each of these rows.
Now, consider a cluster of a \tyb solution of $VG_i$ consisting of two rows $r_1$, $r_2$ of $CVG_i$ and a jolly row
incident on a common docking vertex. Then $8$ entries
($6$ of the vertex blocks, $2$ of the jolly blocks)
are deleted in each of these rows, hence this cluster has
a total cost of $24$. Since the virtual cost of the jolly row is $0$, each of
$r_1$, $r_2$ has virtual cost $12$.
A \tyb solution of $VG_i$ contains four clusters, three clusters containing row incident on the docking
vertices (as described above) and a cluster of three rows incident on $c_{i,6}$, that has
a virtual cost equal to $27$ ($9$ for each row incident on $c_{i,6}$). Each row of $CVG_i$
incident on a docking vertex has a virtual cost of $12$ in a \tyb solution, hence
the total virtual cost of the rows in $CVG_i$ of a \tyb solution is
$27+12\cdot 6=99$.

Consider a $CVG_i$ associated with a \tya solution. Observe that
a \tya solution consists of three cluster, each of cost $27$.
Indeed, each cluster of \tya solution consists of three rows incident on a common vertex.
\end{proof}

In the following we state  two basic results that will be used to show the
L-reduction from MVCC  to $3$-ABP: (i)
each solution $S$ of $3$-ABP can be modified in polynomial
time into a canonical solution $S'$  whose cost is at most that of $S$
(Lemma \ref{cost-can});
(ii)  the graph  $\grafo$ 
has a vertex cover of size $p$ iff the $3$-ABP problem has a canonical solution  of cost  $99 \cdot p +
81\cdot (n-p) +12 m$, (we recall that  $81$ is the total virtual cost of the rows of a \tya solution,
and $99$ is the total virtual cost of the rows of a vertex gadget in a \tyb
solution -- see Theorem \ref{thm:finale}).
We
will first introduce some basic Lemmas that will help in excluding some
possible solutions. 
%

\begin{Lemma}
\label{cluster-dist}
Let $S$ be a solution of an instance of $3$-ABT associated with an
instance of MVCC and let
$C$ be a cluster of $S$ consisting of rows of $CVG_i$. Then
$virt_{S}(r) \geq 9$ for each row $r$ of $C$,  
and $virt_{S}(r) \geq 12$ if $C$ is not a canonical cluster.
\end{Lemma}
\begin{proof}
First notice that by construction $\gadgetgrafo$ does not contain any cycle of
lenght $4$. It follows that if $C$ is not a canonical cluster, then $C$ contains
two rows $e_1$ and $e_2$ not incident on a common vertex. By case 4 of Prop.~\ref{prop:distanze}
$e_1$ and $e_2$ have Hamming distance $12$.

Assume now that $C$ is a canonical cluster, and
let $W$ be the set of vertices incident on the edges of $C$, except for the
center of $C$.
By definition of canonical cluster, $C$ contains  no cycles, moreover $|C|,|W|\ge
3$, therefore for each vertex $v_{i,x} \in W$, there exists one edge in $C$ not
incident in $v_{i,x}$.
Since the vertex encoding $v\text{-}enc_{i,x}$ is applied only to edge
incident in $v_{i,x}$,
the three entries set to $1_b$ by $v\text{-}enc_{i,x}$ are deleted in each row of $C$,
for each $v_{i,x} \in V'$, and the lemma follows.
\end{proof}

\begin{Lemma}
\label{jolly-cluster-dist}
Let $S$ be a solution of an instance of $3$-ABT associated with an
instance of MVCC and let
$C$ be a cluster of $S$ consisting of rows of $VG_i$,
such that $C$ contains at least one jolly row of $VG_i$ and at
least one row of $CVG_i$, then the virtual cost of each non-jolly
row in $C$ is at least $12$.
\end{Lemma}
\begin{proof}
Assume first that $C$ contains exactly one non-jolly row $r_1$.
Then by case 5 of Prop.~\ref{prop:distanze}
$H(r_1,r_j)\geq 5$ for each jolly row $r_j$
and, since $|C|\geq 3$, the cost of $C$ is at least $15$.
Since $r_1$ is the only non-jolly row of $C$,
then the virtual cost of $r_1$ is at least $15$.
Assume now that $C$ contains at least two non-jolly rows, $r_1$, $r_2$ and let
$r_j$ be a jolly row in $C$.
If $C$ contains exactly two rows of $VG_i$ by cases 3,4 of Prop.~\ref{prop:distanze},
$H(r_1,r_2)\geq 6$ 
and by construction there are two positions $h_1, h_2$
in the jolly block  where $r_1[h_z]=r_2[h_z]=0_b$, while $r_j[h_z]=1_b$,
with $z \in \{ 1,2 \}$. Hence the total cost of $C$ is
at least $8|C|$. Since $C$ contains exactly two
non-jolly rows $r_1$, $r_2$, the virtual cost of
$r_1$ and $r_2$ is at least $8|C|/2$.
Since $|C|\geq 3$, the virtual cost of $r_1$ and $r_2$ will be at least
$12$.

Assume that $C$ contains more than two non-jolly rows. Then $C$ contains a
set $C'$, where $C'$ consists of at least $3$ rows of $CVG_i$.
Notice that $C'$ can be a cluster of a feasible solution of $3$-ABT, therefore Lemma~\ref{cluster-dist} applies also to $C'$
and an immediate consequence is that the same $9$ entries in the vertex
blocks must be suppressed also
in the same position in $C$.
Furthermore, by construction, there exist two positions $h_1, h_2$
in the jolly block  where the non-jolly rows have $0_b$, while some of the
jolly rows has value $1_b$. Hence the total cost of $C$ is
at least $11|C|$ and the virtual cost of each non-jolly row in $C$ is at least $11|C|/(|C|-1)$.
But  Remark \ref{Remark:cluster-dim} implies $|C|\leq 5 $,  therefore $11|C|/(|C|-1) \ge 12$
and the virtual cost of each non-jolly row in $C$ is at least
$12$.
\end{proof}

\begin{Lemma}
\label{lem:best-cluster}
Let $S$ be a  solution of an instance of $3$-ABT associated with an
instance of MVCC and let  $C$ be a cluster of $S$ containing a row of $VG_i$. 
Then the virtual cost of each non-jolly row of $C$ is at least $9$.
\end{Lemma}
\begin{proof}
Notice that if $C$ contains a jolly row, then the lemma is a consequence of
Prop.~\ref{prop:distanze} (Cases 2 and 12), Lemmas \ref{cluster-dist} and \ref{jolly-cluster-dist}.
Hence assume that $C$ contains no jolly row.
If  $C$ contains at least two edge gadgets then, by case 9 of
Prop.~\ref{prop:distanze}, the virtual cost of each non-jolly row of $C$ is at
least $18$, therefore we can assume that there is exactly one edge gadget in $C$.

If  $C$ is not a  canonical cluster, there are two rows that not incident on a
common vertex, therefore by cases 1, 4, 8 of Prop.~\ref{prop:distanze} and by
construction of $VG_i$,
each non-jolly row of $C$ has a virtual cost of at least $12$.
The final case that we have to consider for $C$ is when $C$ contains an edge
gadget and two edges of a core vertex gadget and all edges are incident on a
common vertex: in this case we can apply case 7 of Prop.~\ref{prop:distanze}
to obtain the lemma.
\end{proof}

An immediate consequence of Lemma~\ref{lem:best-cluster} and of the construction
of $VG_i$, is that a \tya solution is the optimal solution for the rows associated
with edges of $VG_i$.






\begin{Lemma}
\label{cost-edge-clustered-1-appendix}
Let $S$ be a  solution of an instance of $3$-ABT associated with an
instance of MVCC and let $C$ be a cluster of $S$ containing exactly two edge gadgets $EG_1$ and
$EG_2$. Then each of the virtual costs $virt_S(EG_1)$, $virt_S(EG_2)$ is
at least $21$. If the edge gadgets are not incident on a common
vertex gadget, then $virt_S(EG_1), virt_S(EG_2)\ge 27$.
\end{Lemma}
\begin{proof}
Notice that all $1_b$s in a vertex block of $EG_1$ correspond to
$0_b$s of $EG_2$. The same fact holds for $3$ $1_b$s in
the edge block of $EG_1$, if $EG_1$ and $EG_2$ are incident on a common vertex gadget,
otherwise $6$ $1_b$s in the edge block of $EG_1$ are deleted.
By symmetry of $EG_1$, $EG_2$ the number of deleted columns is at least $18$,
when $EG_1$ and $EG_2$ are incident on a common vertex gadget, otherwise
the number of deleted columns is at least $24$.

Let $r_3 \in C$ different from $EG_1$, $EG_2$.
Since $r_3$ is not an edge gadget, there is a vertex block of
$r_3$ containing $6$ $1_b$s, while both $EG_1$
and $EG_2$ have at most $3$ $1_b$ in that block. It is immediate to notice from
the construction of $EG_1$ and $EG_2$ that this fact leads to at least $3$
additional columns that must be deleted. 
\end{proof}

\begin{Lemma}
\label{prop:cost-2-edges-and-1-vertex-appendix}
Let $S$ be a  solution of an instance of $3$-ABT associated with an
instance of MVCC and let $C$ be a cluster of $S$ containing two edge gadgets $EG_1$,
$EG_2$ that are not incident on a common vertex gadget and a row
$r$ belonging to a vertex gadget, such that $r$ is not adjacent to
$EG_1$ nor to $EG_2$. Then $virt_S(EG_1), virt_S(EG_2), virt_S(r)\geq 30$.
\end{Lemma}
\begin{proof}
Since  $EG_1$ and $EG_2$ are not incident on a common vertex
gadget, by case 10 of  Prop.~\ref{prop:distanze} we know that  $H(EG_1,
EG_2) \geq 24$. Now consider the row $r$ and assume w.l.o.g. that
$r$ belongs to vertex gadget $VG_i$. Since $r$ is not adjacent to
$EG_1$ nor to $EG_2$ there are at least $6$ positions in the $i$-th
block, where $EG_1$ and $EG_2$ have both value $0_b$, while $r$
has value $1_b$.
\end{proof}

\begin{Lemma}
\label{cost-edge-clustered-2-appendix}
Let $S$ be a  solution of an instance of $3$-ABT associated with an
instance of MVCC and let $C$ be a cluster of $S$ containing three edge gadgets $EG_1$, $EG_2$,
$EG_3$. Then each of the virtual costs $virt_S(EG_1)$, $virt_S(EG_2)$,
$virt_S(EG_3)$ is at least $27$. If there is no pair of edge gadgets
in $\{EG_1,EG_2, EG_3 \} $ incident on a common vertex gadget,
then $virt_S(EG_1), virt_S(EG_2), virt_S(EG_3)\geq 36$.
\end{Lemma}
\begin{proof}
Observe that $EG_1$, $EG_2$, $EG_3$ have minimum virtual cost
either when they are all incident on the same vertex gadget or the
set of vertex gadgets to which $EG_1$, $EG_2$, $EG_3$ are incident
consists of three vertex gadgets (that is $EG_1$, $EG_2$, $EG_3$ encode a
cycle of length $3$). Therefore $9$ entries
for each of $EG_1$, $EG_2$, $EG_3$ are deleted in the edge block,
while $18$ entries of the vertex blocks are deleted for each of
$EG_1$, $EG_2$, $EG_3$, since $EG_1$, $EG_2$, $EG_3$ represent
edges incident on a set of $6$ different docking vertices. Hence
the virtual cost of each $EG_i$, $i= \{1,2,3 \} $, is at least
$27$.

Observe that when there is no pair of edge gadgets in
$\{EG_1,EG_2, EG_3 \} $ incident on the same vertex gadget, the
positions of the edge block with value $1_b$ in $EG_1$, $EG_2$,
$EG_3$ are all different, hence at least $18$ entries are
deleted for each of  $EG_1$, $EG_2$, $EG_3$. Hence the virtual
cost of each $EG_i$, $i= \{1,2,3 \} $, is at least $36$.
\end{proof}

\begin{Lemma}
\label{cost-edge-clustered-more-3-appendix}
Let $S$ be a  solution of an instance of $3$-ABT associated with an
instance of MVCC and let $C$ be a cluster of $S$ containing more than three edge
gadgets. Then the virtual cost of each edge gadget in $S$ is at
least $36$.
\end{Lemma}
\begin{proof}
Consider $4$ edge gadgets in $C$: $EG_1$, $EG_2$, $EG_3$, $EG_4$.
First observe that $24$ entries of the vertex blocks are deleted
for each of  $EG_1$, $EG_2$, $EG_3$, $EG_4$, since $EG_1$, $EG_2$,
$EG_3$, $EG_4$ represent edges incident on a set of $8$ different
docking vertices.

A simple argument shows that at least two of such edge gadgets are
not incident on a common vertex gadget. Indeed, 
the set of vertex gadgets on which $EG_1$, $EG_2$, $EG_3$, $EG_4$ are incident
contains at least four vertex gadgets, for otherwise two edge
gadgets must be incident on the same two vertex gadgets.
Hence $12$ entries of the edge block will be deleted from each row in
$C$.
\end{proof}

\begin{Lemma}
\label{cost-edge-vert-diff}
Let $S$ be a  solution of an instance of $3$-ABT associated with an
instance of MVCC and let $C$ be a cluster of $S$ containing an edge gadget $EG_{ij}$ incident on
vertex gadgets $VG_i$ and $VG_j$, two rows $r_x$, $r_y$ adjacent to
$EG_{ij}$, where $r_x$ belongs to $VG_i$ and $r_y$ belongs to
$VG_j$. Then the cost of $C$ in $S$ is at least $18 |C|$.
\end{Lemma}
\begin{proof}
It is an immediate consequence of case 1 of  Prop.~\ref{prop:distanze}.
\end{proof}

\begin{Lemma}
\label{cost-edge-jolly}
Let $S$ be a  solution of an instance of $3$-ABT associated with an
instance of MVCC and let $C$ be a cluster of $S$  containing an edge gadget $EG_{ij}$ and a jolly row
$j_i$. Then $virt_S(EG_{ij})\ge 18$.
\end{Lemma}
\begin{proof}
Observe that by case 2 of  Prop.~\ref{prop:distanze} $H(EG_{ij},j_i) \geq 14$.
Let $\bar{J}$ be the subset of $C$ consisting of all rows of $C$ that are not
jolly rows.
Moreover,  we can assume that $|\bar{J}|\leq 4 $, as $|C|\le 5$.
If $|s(j)|=4$, then there is at least one row $r$ in $\bar{J}$
such that there exist at least $3$ positions where $EG_{ij}$ and
$j_i$ have the same value, while $r$ has a different value.
Hence the virtual cost of $EG_{ij}$ is at least $\frac{17|C|}{|\bar{J}|}>18$.
If $|s(j)| \leq 3$, then the virtual cost of $EG_{ij}$ is at least
$\frac{14|C|}{|\bar{J}|}\ge 18$ and the lemma holds.
\end{proof}

The following Lemma \ref{lem:edge-bound-cluster} is a consequence of cases 1,
2, 7, 8 of
Prop.~\ref{prop:distanze} and the construction of the gadget graph.

\begin{Lemma}
\label{lem:edge-bound-cluster}
Let $S$ be a  solution of an instance of $3$-ABT associated with an
instance of MVCC and let $C$ be a cluster of $S$ with at least an edge gadget $EG_{ij}$.
Then $virt_{S}(EG_{ij})\geq 12$.
\end{Lemma}

Now, we will show our key transformation of a generic solution into a
canonical solution without increasing its cost.
The proof is based on the fact that, whenever
a solution $S$ is not a canonical one, it can be transformed into
a canonical one by applying Alg.~\ref{alg:reduction}.

Let us denote by $S_1$ and $S_2$ respectively, the solution
before and after applying  Alg.~\ref{alg:reduction}. Observe that, by
construction,  in the solution $S_2$
computed by Alg.~\ref{alg:reduction} each edge gadget belongs to a  \tyb solution.

\begin{algorithm}[ht!]
\KwData{a solution $S_1$ consisting of the set  $\{C_1, \cdots,
C_k\}$ of clusters}
Unmark all  edge gadgets and all vertex gadgets;\\
$S_2\gets \emptyset$;\\

\While{there is a cluster $C$ in $\{C_1, \cdots, C_k\}$ with
an unmarked edge gadget} {
  $U(C)\gets$ the set of unmarked edge gadgets in $C$;\\
  $V(C)\gets$ a smallest possible set of vertex gadgets such that each
   edge  gadget in $U(C)$ has at least one endpoint in $V(C)$;
\tcc{$|C| \leq 5$ by Remark \ref{Remark:cluster-dim}, hence we can compute $V(C)$ in polynomial
time. We assume that if a cluster $C$ contains only an edge gadget $EG_{i,j}$
and rows of vertex gadget $VG_i$, then $V(C)= \{ VG_i\}$.}
$E'\gets$ all unmarked edge gadgets incident on some vertices of $V(C)$;\\
Add to $S_2$ a \tyb solution for all vertex gadgets in $V(C)$ and all  edge
gadgets in $E'$;\\
\tcc{Notice that $E' \supseteq U(C)$}
Mark the edge gadgets in $E'$ and the vertex gadgets in $V(C)$;\\
}
Add to $S_2$ a \emph{type a} solution for each umarked vertex gadget;\\
\Return{$S_2$}
\caption{ComputeCanonical($S_1$)}
\label{alg:reduction}
\end{algorithm}

Notice that at the end of the execution of the algorithm, each vertex gadget
is assigned either a \tya or a \tyb solution, and that each row is assigned
to one of those solutions. As \tya solutions are optimal,
we can concentrate on \tyb solutions.

Clusters corresponding to \tyb solutions are built iteratively at line 3-9 of
Alg.~\ref{alg:reduction}. More precisely at each iteration the algorithm  examines a set
of clusters $C_1, \cdots, C_l$ of $S_1$ and it extracts a cluster $C$
containing at least one unmarked edge gadget.
Then the algorithm imposes a \tyb solution on a
set $V(C)$ of vertex gadgets and on a set $E'$ of edge gadgets
so that $U(C) \subseteq E'$.
Such step is the only one where  the virtual cost of some
rows can be modified.
More precisely only edges of the core vertex gadgets
in $V(C)$ and edge gadgets in $E'$ may have in $S_2$ a virtual cost different from
that in $S_1$.

Notice that by Lemma~\ref{lem:edge-bound-cluster},
each edge gadget in $E' - U(C)$, has virtual cost of at least
$12$ in solution $S_1$, and virtual cost $12$ in solution $S_2$. Hence
the virtual cost of such edge gadget is not increased by Alg.~\ref{alg:reduction}.

Now, we consider the rows associated with core vertex gadgets
in $V(C)$ and edge gadgets in $U(C)$.
For simplicity's sake, let us denote by $virt_{S}(V(C))$
the sum of the virtual cost of the set of rows associated with core vertex gadgets in
$V(C)$ in a solution $S$ and similarly,
let us denote by $virt_{S}(U(C))$ the sum of the
virtual cost of the set of rows associated with unmarked edge gadgets in
$U(C)$.

Observe that, by construction, the sets $U(C)$ considered in different
iterations of Alg.~\ref{alg:reduction} are pairwise disjoint, therefore it
makes sense to analyze each iteration separately.
Consequently, it is immediate to conclude that
the correctness of the Alg.~\ref{alg:reduction} can be proved by showing that,
in a generic iteration of Alg.~\ref{alg:reduction}, the following lemma holds:

\begin{Lemma}
\label{claim_eq}
Let $C$ be the cluster containing an unmarked edge gadget found at line $3$ of
Alg.~\ref{alg:reduction}.
Then
$virt_{S_2}(V(C)) - virt_{S_1}(V(C)) \leq virt_{S_1}(U(C)) - virt_{S_2}(U(C))$.
\end{Lemma}
\begin{proof}
The proof consists of several cases. For each case it will suffice to determine
that one of the following conditions hold:

\begin{enumerate}

\item[(i)] $virt_{S_1}(U(C)) -virt_{S_2}(U(C))\ge 18|V(C)|$;

\item[(ii)] $virt_{S_1}(V(C)) + virt_{S_1}(U(C)) \ge 99|V(C)|+12|U(C)|$;

\item[(iii)] $virt_{S_2}(V(C)) - virt_{S_1}(V(C)) \leq virt_{S_1}(U(C)) - virt_{S_2}(U(C))$.
\end{enumerate}

Notice that conditions (i) and (ii) imply condition (iii).
First we show that condition (i) implies condition (iii). Assume that condition
(i) holds, that is $virt_{S_1}(U(C)) -
virt_{S_2}(U(C))\ge 18|V(C)|$. Solution $S_2$ builds a \tya or a \tyb
solution for $V(C)$ (whose cost is at most $99$ for each vertex gadget), while we know
that the optimal solution for $V(C)$ has cost at least $81$
(Lemma~\ref{lem:best-cluster}), which implies that
$virt_{S_2}(V(C)) -
virt_{S_1}(V(C))\le 18|V(C)|$, and hence
$virt_{S_2}(V(C)) - virt_{S_1}(V(C)) \leq virt_{S_1}(U(C)) - virt_{S_2}(U(C))$.

Now we show that condition (i) implies condition (iii)
Assume that conditions (ii) holds, that is
$virt_{S_1}(V(C)) + virt_{S_1}(U(C)) \ge 99|V(C)|+12|U(C)|$.
As by construction of solution $S_2$
$virt_{S_2}(V(C)) + virt_{S_2}(U(C)) = 99|V(C)|+12|U(C)|$, it follows that
$virt_{S_2}(V(C)) - virt_{S_1}(V(C)) \leq virt_{S_1}(U(C)) - virt_{S_2}(U(C))$.

We will distinguish several cases, depending on
the structure of $U(C)$. Recall that Remark \ref{Remark:cluster-dim} implies
$|C|\leq 5$, hence $|U(C)|\leq 5$.
Notice also that, by construction, $|V(C)|\leq |U(C)|$,
and that
$virt_{S_1}(V(C)) - virt_{S_2}(V(C)) \le 18|V(C)|$
as the set of rows of each vertex gadget in $V(C)$ has a
total cost of at least $81$ in solution $S_1$ and at most $99$ in
solution $S_2$.

\begin{itemize}
\item
Assume that $|U(C)|>3$.
By Lemma~\ref{cost-edge-clustered-more-3-appendix}, the virtual cost (in $S_1$) of each edge gadget in $U(C)$ is at least  $36$.
Therefore $virt_{S_1}(V(C)) + virt_{S_1}(U(C)) \ge  81|V(C)| + 36|U(C)|
=81|V(C)|+ 12|U(C)|+ 24|U(C)|  \geq 81|V(C)| + 12|U(C)|+ 24|V(C)|
= 12|U(C)|+ 105|V(C)| \geq  12|U(C)| + 99 |V(C)|$, as required by condition (ii).
\item
Assume that $|U(C)|=3$  and no two gadgets in $U(C)$ are
incident on a common vertex gadget.
By Lemma~\ref{cost-edge-clustered-2-appendix}, the virtual cost (in $S_1$) of each edge
gadget in $U(C)$ is at least of $36$.
We can apply the same analysis of case $|U(C)|>3$ to show that
$virt_{S_1}(V(C)) + virt_{S_1}(U(C)) \ge 12|U(C)|+ 105|V(C)| \geq  12|U(C)| + 99 |V(C)|$,
as required by condition (ii).
\item
Assume that $|U(C)|=3$  and two gadgets in $U(C)$ are
incident on a common vertex gadget.
By Lemma~\ref{cost-edge-clustered-2-appendix}, the virtual cost (in $S_1$) of each edge
gadget in $U(C)$ is at least of $27$, but notice that $|V(C)|\leq 2$.
Therefore $virt_{S_1}(V(C)) + virt_{S_1}(U(C)) \ge 27\cdot 3+ 81|V(C)|$.
It is immediate to notice that  $27\cdot 3+ 81|V(C)| \ge 12\cdot 3 + 99|V(C)|$
when $|V(C)|\leq 2$, as required by condition (ii).
\item
Assume that $|U(C)|=2$ and such two gadgets $EG_1$, $EG_2$ in $U(C)$ are not
incident on a common vertex gadget.
By Lemma~\ref{cost-edge-clustered-1-appendix}, the virtual cost in $S_1$ of each edge
gadget in $U(C)$ is at least $27$, while by Prop. \ref{prop-cost-canonical-sol} the virtual cost in $S_2$ of each edge
gadget in $U(C)$ is exactly $12$. Hence $virt_{S_1}(U(C)) - virt_{S_2}(U(C)) \geq 30$.
If $|V(C)|=1$ then $virt_{S_2}(V(C)) - virt_{S_1}(V(C)) \leq 18$ and, a
fortiori, $virt_{S_2}(V(C)) - virt_{S_1}(V(C)) \leq virt_{S_1}(U(C)) - virt_{S_2}(U(C))$,
as required by condition (iii).
Therefore we are only interested in the case $|V(C)|\geq 2$. Since $|V(C)|\leq
|U(C)|= 2$, we can assume that $|V(C)|=2$ and, consequently,
$virt_{S_2}(V(C)) - virt_{S_1}(V(C)) \leq 36=18|V(C)|$ (this observation will
be used in the reamining part of this case).

Let $r$ be a row in $C$ which is not an edge gadget (one must exist because
$|C|\geq 3$ and there are exactly two edge gadgets in $C$).
We have to distinguish two cases, according to the fact that $r$ is a jolly
row or a row of a vertex gadget.

First consider the case when $r$ is a jolly row.
By Lemma~\ref{cost-edge-clustered-1-appendix} at least $27$ entries of the vertex blocks
in the rows of $C$ are deleted. Since $r$ is a jolly row $virt_{S_1}(EG_1), virt_{S_1}(EG_2)\ge
27|C|/(|C|-1)\ge 33$, as $|C|\le 5$.
Therefore $virt_{S_1}(U(C))-virt_{S_2}(U(C))=
virt_{S_1}(EG_1) + virt_{S_1}(EG_2) - virt_{S_2}(EG_1) - virt_{S_2}(EG_2)\ge
33\cdot 2 - 12\cdot2 = 42$, which implies that
$virt_{S_2}(V(C))+virt_{S_2}(U(C)) \le virt_{S_1}(V(C))+ virt_{S_1}(U(C))$,
as required by condition (iii).


Assume now that $r$ is a row of a vertex gadget $X$, we have to consider two
subcases depending on the fact that $r$ is adjacent or not to $EG_1$ or
$EG_2$.
Assume that $r$ is  not adjacent to $EG_1$ nor to $EG_2$.
By Lemma~\ref{prop:cost-2-edges-and-1-vertex-appendix}, $virt_{S_1}(U(C)),
virt_{S_2}(U(C))\ge 30$.
As we have assumed that $|V(C)=2|$ and $virt_{S_2}(V(C)) - virt_{S_1}(V(C)) \leq
36$, we can immediately prove that
$virt_{S_2}(V(C))+virt_{S_2}(U(C)) \le virt_{S_1}(V(C))+ virt_{S_1}(U(C))$,
as required by condition (iii).
The last subcases that we have to consider is when $r$ is adjacent to $EG_1$ or
$EG_2$. Assume w.l.o.g. that $r$ is adjacent to $EG_1$ and that $X \subseteq V(C)$.
Observe that
$r$ is co-clustered in $S_2$ in a \tya or in a \tyb solution, hence by Prop. \ref{prop-cost-canonical-sol}
$virt_{S_2}(r)\le 12$, while $virt_{S_1}(r)\ge 27$, as $virt_{S_1}(EG_1)\ge
27$ by Lemma~\ref{cost-edge-clustered-1-appendix}. Taking into account
the fact that $virt_{S_1}(r)- virt_{S_2}(r)\ge 15$,
as $r$ is a row of $X \subseteq V(C)$,
we immediately obtain that
$virt_{S_2}(V(C)) - virt_{S_1}(V(C))\le 18\cdot 2-15=21$ which can be coupled with
$virt_{S_1}(U(C)) - virt_{S_2}(U(C)) \geq 30$ proved before to obtain
$virt_{S_2}(V(C))+virt_{S_2}(U(C)) \le virt_{S_1}(V(C))+ virt_{S_1}(U(C))$,
as required by condition (iii).
\item
Assume that $|U(C)|=2$  and that the edge gadgets $EG_1$, $EG_2$ in $U(C)$ are
incident on a common vertex gadget. Hence, by construction of the algorithm,
$|V(C)|\leq 1$, therefore $virt_{S_2}(V(C)) - virt_{S_1}(V(C)) \le 18$.
By Lemma~\ref{cost-edge-clustered-1-appendix}, the virtual cost (in $S_1$) of each edge
gadget in $U(C)$ is $21$, therefore
$virt_{S_1}(U(C)) - virt_{S_2}(U(C)) \ge (21-12)\cdot2 = 18$,
as required by condition (i).
\item
Assume that $U(C)=\{EG_{l,h}\}$ and $EG_{l,h}$ is clustered with two
rows $r_1$, $r_2$ from different vertex gadgets $VG_i$, $VG_j$.
Since $|U(C)|=1$, $|V(C)|\leq 1$, hence
$virt_{S_2}(V(C)) - virt_{S_1}(V(C)) \le 18$.
By Lemma~\ref{cost-edge-vert-diff},
$virt_{S_1}(EG_{l,h}), virt_{S_1}(r_1), virt_{S_1}(r_2) \geq 18$,
while by Prop. \ref{prop-cost-canonical-sol} $virt_{S_2}(EG_{l,h})$, $virt_{S_2}(r_1)$, $virt_{S_2}(r_2) \le 12$,
hence
$virt_{S_1}(EG_{l,h})-virt_{S_2}(EG_{l,h})\ge 6$ while
$virt_{S_2}(V(C)) - virt_{S_1}(V(C)) \le 6$ which immediately implies
$virt_{S_2}(V(C))+virt_{S_2}(U(C)) \le virt_{S_1}(V(C))+ virt_{S_1}(U(C))$,
as required by condition (iii).
\item
Assume that $U(C)=\{EG_{l,h}\}$ and that $C$ contains
the edge gadget $EG_{l,h}$ and (at least) two rows $r_1$, $r_2$ of $VG_j$,
with $j \neq i,l$.
Since $|U(C)|=1$, $|V(C)|\leq 1$, hence
$virt_{S_2}(V(C)) - virt_{S_1}(V(C)) \le 18$.

Initially we will prove that  for each row $r$ of $C$ the
virtual cost $virt_{S_1}(r) \geq 21$.
In fact $r_1$ and $r_2$  may share a gadget encoding operation and a vertex
encoding operation (if they are incident on a common vertex), therefore there
are three distinct gadget encoding operations and four distinct gadget
encoding operations overall, and none of such operation is shared by all edges
in $C$. An immediate consequence is that at least $21$ entries of each row of
$C$ must be suppressed, therefore $virt_{S_1}(\{r_1,r_2,EG_{l,h}\}
)-virt_{S_2}(\{r_1,r_2,EG_{l,h}\} )\ge (21-12)\cdot 3=27$,
as by Prop. \ref{prop-cost-canonical-sol} each
row of $C$ in $S_2$ has virtual cost at most $12$.
For bookkeeping purposes we attribute the entire value of $virt_{S_1}(\{r_1,r_2,EG_{l,h}\}
)-virt_{S_2}(\{r_1,r_2,EG_{l,h}\} )$ to $EG_{l,h}$ (this bookkeeping trick is
possible as a row $r$ is allowed to  give its ``credit'' only to an edge
gadget with which it is co-clustered in $S_1$), therefore we obtain  $virt_{S_1}(U(C)) -
virt_{S_2}(U(C))\ge 27\ge 18$, as required by condition (i).
\item
Assume that $U(C)=\{EG_{i,l}\}$ and that $C$ contains
the edge gadget $EG_{i,l}$ and (at least) two rows $r_1$, $r_2$ of $VG_l$, with
$r_2$  not adjacent
to $EG_{i,l}$.
Notice that, by case 8 of Property~\ref{prop:distanze}, $H(r_2, EG_{i,l}) \geq 15$.
Therefore the $virt_{S_1}(\{EG_{i,l}, r_1, r_2\})\ge 45$, while
$virt_{S_2}(\{EG_{i,l}, r_1, r_2\})\le 36$ by Prop. \ref{prop-cost-canonical-sol}.

In what follows we will consider the virtual costs
of the rows in $VG_i$. More precisely, let $T$ be the set $VG_i -
\{r_1,r_2\}$;
we will show that there exists a row $r\in T$
such that $virt_{S_1}(r)\ge 12$.

Assume initially that there exists a row $r\in T$ that is clustered
with an edge gadget or with a row belonging to a  different vertex gadget
$VG_j$. Then $virt_{S_1}(r)\ge 12$ by Lemma~\ref{lem:edge-bound-cluster} and
case 1 of Prop.\ref{prop:distanze}.
Hence, we can assume that the rows in $T$ are clustered only
with rows of $VG_i$, which implies that $r_1$ and $r_2$ are not clustered
together with any row in $T$, as $r_1$ and $r_2$ are clustered
with $EG_{i,l}$.


Since $T$ contains $7$ rows, a trivial counting argument shows that $4$ rows
of $T$ are clustered together, or there is a row $r\in T$ that is clustered
with a jolly row of $VG_i$. Indeed, if four rows of $T$ are clustered
together, by construction there are two of those four rows
that are not incident on a common
vertex, therefore an immediate application of case 4 of
Prop.~\ref{prop:distanze} gives the desired result. If $r$ is clustered with a
jolly row then, by Lemma~\ref{jolly-cluster-dist}, $virt_{S_1}(r_w) \geq 12$.

Now we know that there is a row $r$ in  $VG_i -\{r_1,r_2\}$ such that
$virt_{S_1}(r)\ge 12$. Moreover
$virt_{S_1}(\{EG_{i,l}, r_1, r_2\})\ge 45$ and
$virt_{S_2}(\{EG_{i,l}, r_1, r_2\})\le 36$.

By Lemma~\ref{lem:best-cluster} the virtual cost of any row of $VG_i$
different from $r$, $r_1$, $r_2$ is at least $9$.
Since $virt_{S_1}(r_1)=virt_{S_1}(r_2) \geq 15$,
$virt_{S_1}(VG_i)=6 \cdot 9+12+15 \cdot 2= 96$, while $virt_{S_2}(VG_i)=99$
(since $S_2$ has a \tyb solution for the rows in $VG_i$).

Since $virt_{S_1}(EG_{i,j})=15$ and $virt_{S_2}(EG_{i,j})=12$ by Prop. \ref{prop-cost-canonical-sol}, it is immediate
to obtain
that $virt_{S_1}(V(C)) +virt_{S_1}(U(C)) \ge virt_{S_2}(V(C))+virt_{S_2}(U(C))$
as required by condition (iii).
\item
Assume that $U(C)=\{EG_{l,h}\}$ and that $C$ contains at least a jolly row $e_j$
of a vertex gadget $VG_l$ or $VG_h$ (w.l.o.g. $VG_l$).
We assume that
$VG_l \subseteq V(C)$.
Furthermore, we can assume that no other row of a different vertex gadget
belongs to $C$ otherwise the previous cases hold.
By case 2 of Prop.~\ref{prop:distanze} the Hamming distance of $EG_{l,h}$ and
$e_j$  is at least $14$, therefore $virt_{S_1}(C)\geq 14|C|$, while
$virt_{S_2}(C) \leq 12(|C|-1)$, hence
$virt_{S_1}(C)- virt_{S_2}(C)\geq 2(|C|-1)+14$ and, since $|C|\geq 3$,
$virt_{S_1}(C)- virt_{S_2}(C)\geq 18$.
Once again,
we attribute the entire value of $virt_{S_1}(C)- virt_{S_2}(C)\geq 18$ to
$EG_{l,h}$ (this bookkeeping trick is
possible as a row $r$ is allowed to  give its ``credit'' only to an edge
gadget with which it is co-clustered in $S_1$), therefore
$virt_{S_1}(U(C))- virt_{S_2}(U(C))\geq 18$, as required
by condition (i).
\end{itemize}

The proof is completed by the observation that the only possible case that is
not explicitly considered in the above cases is when an unmarked edge gadget
is clustered in $S_1$ only with rows of the core vertex gadget $VG_i$ and all rows
share a common vertex. In such case, Alg.~\ref{alg:reduction} does not modify
the clustering, as $V(C)$ is made only by the vertex gadget $VG_i$ and $V(C)$
has a \tyb solution in $S_2$.
\end{proof}

\begin{Lemma}
\label{cost-can}
Let $S$ be a  solution of an instance of $3$-ABT associated with an
instance of MVCC, then
we can compute in polynomial time a canonical solution $S_c$
such that $c(S_c) \leq c(S)$.
\end{Lemma}

\begin{Theorem}
\label{thm:finale}
Let $\grafo=(V,E)$ be an instance of MVCC. Then $\grafo$ has a
cover of size $p$ if and only if the corresponding instance $R$ of $3$-ABT
has 
a (canonical) solution $S$ of cost $99p+81(n-p)+12m$.
\end{Theorem}
\begin{proof}
Let us show that if $\grafo$ has a vertex cover $V_c$ of size $p$, then
$R$ has a solution $S$ of cost $99p+81(n-p)+12m$.
Since $V_c$ is a vertex cover then it is possible to  construct a
canonical solution $S$ for $R$ consisting of  a \tyb solution for
all vertex gadgets associated with vertices in $V_c$ and   a \tya
solution for all other vertex gadgets.
Indeed each edge gadget can be clustered  in a \tyb solution of a
vertex gadget to which the edge is incident, choosing arbitrarily
whenever there is more than one possibility. Finally, for each
docking vertex, its jolly rows that are  not used in some \tyb
solution are clustered together. The cost derives immediately by
previous observations.

Let us consider now a solution $S$ of $3$-ABP over instance $R$ with cost
$99p+81(n-p)+12m$. By Lemma~\ref{cost-can} we can assume that $S$
is canonical solution, therefore $R$ has a set $V_c$ of $p$ vertex
gadgets that are associated with a \tyb solution. By construction,
each edge gadget must be in a \tyb solution, for otherwise $S$ is a not
canonical solution. Hence the set of  vertices of $\grafo$ associated
with vertex gadgets in $V_c$ is a vertex cover of $\grafo$ of size $p$.
\end{proof}

Since the cost of a canonical solution of $3$-ABP and the size of a vertex
cover of the graph $\grafo$ are linearly related, the reduction is an L-reduction,
thus completing the proof of APX-hardness.

\section{APX-hardness of $4$-AP($8$)}
\label{sec-APX-8-4}

In this section we prove that the $4$-anonymity problem is APX-hard even if
all rows of the input table have $8$ entries (this restriction is denoted by
$4$-AP($8$)). 
More precisely, we give an L-reduction from Minimum Vertex Cover on Cubic
Graphs (MVCC) to $4$-AP($8$).

Given a cubic graph $\grafo=(V,E)$, with $V= \{ v_1, \dots, v_n \}$ and
$E= \{ e_1 , \dots, e_m \}$, we will construct an instance $R$
of $4$-AP($8$) consisting of
a set $R_i$ of $5$ rows  for each vertex $v_i \in V$,
an \emph{edge row} $r(i,j)$  for each edge $e=(v_i, v_j) \in E$ and a set $F$ of $4$ rows.
The 8 columns are divided in 4 blocks of two columns each.
For each vertex $v_i$,  all the rows in $R_i$  have associated a block
called edge block, denoted as $b(R_i)$, so that
$b(R_i) \neq b(R_j)$ for each $v_j$ adjacent to $v_i$ in $\grafo$.
The latter property can be easily enforced in
polynomial time as the graph is cubic.

%
%
The entries of the rows in $R_i=\{ r_{i,1}, \dots , r_{i,5}  \}$, 
are over the alphabet $\Sigma(R_i)= \{ a_{i,1},\dots, a_{i,5}, a_i  \}  $.
The entries of the columns corresponding to the edge block  $b(R_i)$, as well
as to the odd columns are set to $a_i$ for all the rows in $R_i$.
The entries of the even columns not in $b(R_i)$
of each row $r_{i,h}$ are set to $a_{i,h}$.

For each edge $e=(v_i, v_j)$, we define a row $r(i,j)$ (called \emph{edge row}) of $R$.
Row $r(i,j)$ has
value $a_i$ (equal to the values of the rows in $R_i$)
in the two columns corresponding to the edge block $b(R_i)$,
value $a_j$ (equal to the values of the rows in $R_j$)
in the two  columns corresponding to the edge block  $b(R_j)$, and value
$t_{i,j}$ in all other columns.
Given a set of rows
$R_i$, we denote by $E(R_i)$ the set of rows $r(i,j)$, $r(i,l)$, $r(i,h)$,
associated with edges of $\grafo$ incident in $v_i$.
Finally, we introduce in the instance $R$ of $4$-AP($8$)
a set of $4$ rows $F=\{f_1, f_2, f_3, f_4 \}$,
over alphabet $\Sigma(F)=  \{ u_1 \dots, u_4 \}$. Each row $f_i$ is called a
\emph{free row} and all its entries have value $u_i$.
Since all tables have $8$ entries, w.l.o.g. we can assume that there exists only one cluster $F_c$, called the
\emph{filler cluster}, whose cost is equal to $8|F_c|$. In fact, if
there exists two clusters $F_c$, $F'_c$ exist,
whose cost is equal to $8|F_c|$ and $8|F'_c|$ respectively,
then we can merge them without increasing the cost of the solution.
The free rows
must belong to $F_c$, as each free row has Hamming distance $8$
with all other rows of $R$.
Notice that, by construction, $\Sigma(R_i) \cap \Sigma(R_j)=\emptyset$, hence two rows have Hamming distance
smaller than $8$ only if they both belong to $R_i\cup E(R_i)$ for some $i$.
This observation immediately implies the
following proposition.

\begin{Proposition}
\label{sol-B-AT-1}
Let $S$ be a  solution of an instance of $4$-AP($8$) associated with an
instance of MVCC and let $C$ be a cluster of $S$ where
each row in $C$ has cost strictly less than $8$.
Then $C \subseteq R_i \cup E(R_i)$.
\end{Proposition}

Since in $R_i\cup E(R_i)$ there are $8$ rows, there can be at most two sets having rows in
$R_i\cup E(R_i)$ and satisfying the statement of Proposition~\ref{sol-B-AT-1}.
Consider a solution $S$ and a set of rows $R_i$. We will say that $S$ is a
\emph{black} solution for $R_i$ if in $S$ there is a cluster
containing $4$ rows of $R_i$ and a cluster containing
one row of $R_i$ and the three rows of $E(R_i)$.
We will say that $S$ is a  \emph{red} solution for $R_i$ if in $S$ there is a
cluster consisting of all $5$ elements of $R_i$. By an abuse of language we
will say respectively that $R_i$ is black (resp. red) in $S$.
Given an instance $R$ of $4$-AP($8$),
a solution where each set $R_i$ is either black or red is called a \emph{canonical solution}.
Notice that a canonical solution consists of a filler cluster and a red or black solution
for each $R_i$.
The main technical step in our reduction consists of proving
Lemma~\ref{sol-B-AT-2}, which states that, starting from a solution $S$, it is
possible to compute in
polynomial time a canonical solution $S'$ with cost not larger than that of $S$.
To achieve such goal we need some technical results.


Next we show that moving the rows of $R_i$ that are in the filler cluster to another existing
cluster that contains some rows of $R_i$ (if possible) or to a new cluster,
does not increase the cost of the solution.

\begin{Lemma}
\label{sol-C-low-bound}
Let $S$ be a  solution of an instance of $4$-AP($8$) associated with an
instance of MVCC, and let $r$ be a row of $R_i$.  Then at least three even
entries of $r$ that are not in the edge block are deleted.
\end{Lemma}
\begin{proof}
The lemma follows from the property that a row $r \in R_i$ must be co-clustered
with at least three other rows, and that $r$ disagrees with
any other row of $R$ in the even entries not in the edge block of $R_i$.
\end{proof}

\begin{Lemma}
\label{sol-R_i}
Let $S$ be a  solution of an instance of $4$-AP($8$) associated with an
instance of MVCC. Then we can compute in polynomial time a
solution $S'$ with cost not larger than that of $S$ and such that in $S'$
there exist at most two sets containing some rows of $R_i$.
\end{Lemma}
\begin{proof}
Consider a generic set of rows $R_i$; clearly if at most two sets of $S$
contain some rows of $R_i$, then $S$ satisfies the lemma, therefore assume that in
$S$ there are at least three sets containing some rows of $R_i$.
Let $C_i^1 , \dots, C_i^z$ be the clusters of $S$ containing rows of $R_i$.
By a simple counting argument, at most one of the clusters $C_i^j$ (w.l.o.g. let
$C_i^1$ be such cluster), can be a subset of  $R_i \cup E(R_i)$, all
other clusters $C_i^j$ contain some rows not in $R_i \cup E(R_i)$ (as well as
some rows in $R_i \cup E(R_i)$ by construction), therefore
the cost of each row of $C_i^2 , \dots, C_i^z$ is $8$.
Move all the rows of $C_i^2, \dots ,  C_i^z$ to the filler cluster to obtain a
solution whose cost is not larger than that of the original solution. At
the same time the resulting solution has exactly two clusters containing some
rows of $R_i$. Repeating the process for all sets $R_i$ completes the proof.
\end{proof}

Hence, in what follows we assume that in any solution there are
at most two sets containing rows of each set $R_i$.

\begin{Lemma}
\label{lemma:small-filler}
Let $S$ be a  solution of an instance of $4$-AP($8$) associated with an
instance of MVCC. Then it is possible to compute in
polynomial time a solution $S'$, whose cost is not larger than that of $S$, such that
the filler cluster $F_c$
of $S'$ consists of all free rows  and some (possibly zero) edge rows.
Moreover in $S'$ there are at most two clusters containing rows of $R_i$.
\end{Lemma}
\begin{proof}
Consider a generic set $R_i$. By Lemma~\ref{sol-R_i} we already know that
there are at most two clusters of $S$ containing some rows of $R_i$. Assume
initially that there exists only one cluster of $S$ containing some rows of
$R_i$. If such cluster is the filler cluster, then remove all rows of $R_i$
from the filler cluster and make $R_i$ a cluster of $S'$. In the resulting
solution none of the rows of $R_i$ are in the filler cluster.

Consider now the case that
there are exactly two clusters $C_1$, $C_2$ of $S$ containing some rows of $R_i$.
If one of those clusters (say $C_2$) is the filler cluster $F_c$, then move all rows
of $R_i$ that are in $F_c$ to $C_1$, obtaining two clusters $C_1\cup R_i$ and $F_c-R_i$.
Notice that before moving the rows of $C_2$ to $C_1$, at least one
even position not in $b(R_i)$ is deleted for each row in $C_1$.
As each row moved from $C_2$ to $C_1$ differs from any other rows of $C_1$ in at most
two not yet deleted entries, and all the entries of block $b(R_i)$ are equal for the rows in
$R_i \cup E(R_i)$, it follows that this change
does not increase the cost of the solution.
%
\end{proof}

Now we are ready to prove Lemma \ref{sol-B-AT-2}.

\begin{Lemma}
\label{sol-B-AT-2}
Let $S$ be a  solution of an instance of $4$-AP($8$) associated with an
instance of MVCC. Then it is possible to compute in
polynomial time a canonical solution $S'$ with cost not larger than that of $S$.
\end{Lemma}
\begin{proof}
Consider a generic set of rows $R_i$. By Lemma~\ref{lemma:small-filler} no
row of $R_i$ belongs to the filler cluster. Therefore,
if $R_i$ is neither red nor black in $S$,
then the rows of $R_i$ can be partitioned in $S$ in one of  the following two ways:
(i) a cluster $C_1$ contains three rows of $R_i$ and a
row of $E(R_i)$, while $C_2$  contains two rows of $R_i$ and two rows of
$E(R_i)$, or (ii) a cluster $C$ of $S$ contains all rows of $R_i$ and
some rows of $E(R_i)$.

In the first case, replace $C_1$ and $C_2$ with two clusters $C_1'$, $C_2'$,
where $C_1'$ consists of $4$ rows of $R_i$ and $C_2'$ consists of a row of $R_i$
and all rows of $E(R_i)$ (it is immediate to notice that $C_1'$,
$C_2'$ have cost $12$ and $24$ respectively, while $C_1$ and $C_2$ have both cost $24$).

In the second case let $X$ be the set $C\cap E(R_i)$,
replace $C$ with cluster $C'=R_i$
 and  move all rows in $X$ to the
filler cluster. Let $x=|X|$, then the cost of $C$ in $S$ is
$6 (5 + x)$, while the cost of $C'$ and $X$  in the new solution is equal to
$3\cdot 5 + 8\cdot x$. Since $x\le 3$, the cost of the new solution is
strictly smaller than that of $S$.
\end{proof}

Notice that, given a canonical solution $S$, each red set $R_i$ in $S$ has
a cost of $15$, each black set $R_i$ in $S$ has a cost of $36$
(that is a cost of $12$ associated with the rows of $R_i$ and
a cost of $24$ associated with $3$ edge rows in the black solution of $R_i$),
and the filler cluster $F_c$ has cost $8|F_c|$.
Now, it is easy to see that Lemma \ref{lemma:cost-canonical-solution} holds.

\begin{Lemma}
\label{lemma:cost-canonical-solution}
Let $S$ be a canonical  solution with $k$ red sets $R_i$ of an instance of $4$-AP($8$) associated with an
instance of MVCC. Then $S$ has cost
$12(|V|-k)+15k+8|E|+32$.
\end{Lemma}

Now, we can show that the sets of rows $R_i$ that are red in a canonical solution $S$
corresponds to a cover of the graph $\grafo$.

\begin{Lemma}
\label{lemma:cost-canonical-solution-1}
Let $S$ be a  canonical solution of cost $12(|V|-k)+15k+8|E|+32$ of an instance of $4$-AP($8$) associated with an
instance of MVCC.
Then it can be computed in polynomial time a vertex cover of $\grafo$ of size $k$.
\end{Lemma}
\begin{proof}
Since $S$ is a canonical solution of $4$-AP($8$) of cost $12(|V|-k)+15 k + 8 |E|$,
then all the sets $R_i$ must be associated with either a red or a black solution.
Furthermore, since all the edge rows have a cost of $8$ in $S$,
then there must exist
$k$ sets $R_i$ associated with a red solution, and $|V|-k$ sets associated with a black solution.

Notice that, given two black sets $R_i$ and $R_j$,
there cannot be an edge between two vertices $v_i$ and $v_j$ of $\grafo$
associated with $R_i$ and $R_j$, by definition of black solution. Hence,
the set of vertices associated with black sets of $S$
is an independent set of $\grafo$, which in turn implies that the vertices associated with
red sets are a vertex cover of $\grafo$.
\end{proof}

\begin{Theorem}
\label{thm:fin-th-red2}
The $4$-AP($8$) problem is APX-hard.
\end{Theorem}
\begin{proof}
Let $C$ be a vertex cover of graph $\grafo$. Then, it is easy to see that a
canonical solution $S$ of the  instance of $4$-AP($8$) associated with $\grafo$
such that $S$ has cost at most $12|V|+3|C| +8|E|+32$ can be computed in
polynomial time by defining a black solution
for each set $R_i$ associated with a vertex $v_i \in V \setminus C$, a red solution
for each set $R_i$ associated with a vertex $v_i \in C$,
and assigning all the remaining rows to the filler cluster $F_c$.

On the other side, by Lemma \ref{lemma:cost-canonical-solution-1}, starting from a canonical
solution of $4$-AP($8$) with size  $12(|V|-k)+15k+8|E|+32$, we can compute in polynomial
time a cover of size $k$ for $\grafo$.
%
Since the cost of a canonical solution of $4$-AP($8$) and
the size of a vertex cover of the graph $\grafo$ are linearly related, the reduction is an L-reduction,
thus completing the proof.
\end{proof}



\section{Acknowledgements}

PB and GDV have been partially supported by FAR 2008 grant
  ``Computational models for phylogenetic analysis of gene variations''.
PB has been partially supported by
the MIUR PRIN grant  ``Mathematical aspects and emerging
applications of automata and formal languages''.

\bibliographystyle{abbrv}
\bibliography{abbreviations,graphs,clustering,books,complexity,database}

\end{document}